\documentclass[sigconf,nonacm,10pt]{acmart}

\usepackage[compact]{titlesec} 
\titlespacing{\section}{0pt}{0pt}{0pt}
\titlespacing{\definition}{0pt}{0pt}{0pt}
\usepackage{hyperref}
\usepackage{subcaption}
\usepackage{cleveref}
\usepackage{xspace}
\crefname{figure}{Figure}{Figures}
\usepackage{float}
\usepackage{enumitem}
\usepackage{longtable}
\usepackage{graphicx}
\usepackage{algorithm2e}
\usepackage{balance} 
\usepackage{xcolor}
\usepackage{tabularx}
\usepackage{hyphenat}
\usepackage{amssymb}
\usepackage{amsmath}

\usepackage{lineno}

\usepackage{booktabs} %

\newtheorem{example}{Example}

\newtheorem{definition}{Definition}

\newtheorem{theorem}{Theorem}

\newcommand{\hide}[1]{}
\newcommand{\ourapproach}{\emph{BugDoc}\xspace}

\newcommand{\shortcut}{\emph{Shortcut}\xspace}
\newcommand{\stacked}{\emph{Stacked Shortcut}\xspace}
\newcommand{\debuggingdecisiontrees}{\emph{Debugging Decision Trees}\xspace}

\renewcommand{\paragraph}[1]{\vspace{0.1cm}\noindent \textbf{#1}}

\newcommand{\paragraphemph}[1]{\vspace{0.1cm}\noindent \emph{#1}}

\newcommand{\succeed}{\texttt{succeed}\xspace}
\newcommand{\fail}{\texttt{fail}\xspace}

\newenvironment{myitemize}%
{\begin{list}{$\bullet$}{%
	\setlength{\labelsep}{2pt}\setlength{\leftmargin}{0pt}%
	\setlength{\labelwidth}{0pt}%
	\setlength{\listparindent}{0pt}}}
{\end{list}}

\begin{document}

\title{BugDoc: Algorithms to Debug Computational Processes}

\author{Raoni Louren\c{c}o}
\affiliation{%
  \institution{New York University}
  }
\email{raoni@nyu.edu}

\author{Juliana Freire}
\affiliation{%
  \institution{New York University}
}
\email{juliana.freire@nyu.edu}

\author{Dennis Shasha}
\affiliation{%
  \institution{New York University}
  }
\email{shasha@courant.nyu.edu}

\begin{abstract}
Data analysis for scientific experiments and enterprises, large-scale  simulations, and machine learning tasks all entail the use of complex computational pipelines to reach quantitative and qualitative conclusions. If some of the activities in a pipeline produce erroneous outputs, the pipeline may fail to execute or produce incorrect results. Inferring the root cause(s) of such failures is challenging, usually requiring time and much human thought, while still being error-prone. We propose a new approach that makes use of iteration and provenance to automatically infer the root causes and derive succinct explanations of failures. Through a detailed experimental evaluation, we assess the cost, precision, and recall of our approach compared to the state of the art. Our experimental data and processing software is available for use, reproducibility, and enhancement.
\end{abstract}



\copyrightyear{2020}
\acmYear{2020}
\setcopyright{acmlicensed}\acmConference[SIGMOD'20]{Proceedings of the 2020 ACM SIGMOD International Conference on Management of Data}{June 14--19, 2020}{Portland, OR, USA}
\acmBooktitle{Proceedings of the 2020 ACM SIGMOD International Conference on Management of Data (SIGMOD'20), June 14--19, 2020, Portland, OR, USA}
\acmPrice{15.00}
\acmDOI{10.1145/3318464.3389763}
\acmISBN{978-1-4503-6735-6/20/06}

\maketitle

\section{Introduction} \label{sec:intro}
Computational pipelines are widely used in many domains, from astrophysics and biology to enterprise analytics. They are characterized by interdependent modules,  associated parameters, and data inputs. Results derived from these pipelines lead to conclusions and, potentially, actions.
If one or more modules in a pipeline produce erroneous or unexpected outputs, these conclusions may be incorrect. 
Thus, it is critical to identify the causes of such failures.

Discovering the root cause of failures in a pipeline is challenging because problems can come from many different sources, including bugs in the code, input data, software updates, and improper parameter settings.  Connecting the erroneous result to its root cause is especially difficult for long pipelines or when multiple pipelines are composed.
Consider the following real but sanitized examples.

 \paragraphemph{Example: Enterprise Analytics.} In an application deployed by a major software company, plots for sales forecasts showed a sharp decrease compared to historical values. After much investigation, the problem was tracked down to a data feed (coming from an external data provider), whose temporal resolution had changed from monthly to weekly. The change in resolution affected the predictions of a machine learning pipeline, leading to incorrect forecasts.

\paragraphemph{Example: Exploring Supernovas.} In an astronomy experiment, some visualizations of supernovas  presented unusual artifacts that could have indicated a discovery. The experimental analysis consisted of multiple pipelines run at different sites, including data collection at the telescope site, data processing at a high-performance computing facility, and data analysis run on the physicist's desktop. After spending substantial time trying to verify the results, the physicists found that a bug introduced in the new version of the data processing software had caused the artifacts.

To debug such problems, users currently expend considerable effort reasoning about the  effects of the many possible different settings. This requires them to tune and execute new pipeline instances to test hypotheses manually, which is tedious, time-consuming, and error-prone.

We propose new methods and a system that automatically and iteratively identifies one or more minimal causes of failures in general computational pipelines (or workflows). %

\paragraph{The Need for Systematic Iteration.} %
Consider the example in Figure~\ref{fig:pipeline}, which shows a generic template for a machine learning pipeline and a log of different instances that were run with their associated results.

The pipeline reads a dataset, splits it into training and test subsets, creates and executes an estimator, and computes the F-measure score using 10-fold cross-validation.  
A data scientist uses this template to  understand how different estimators perform for different types of input data, and ultimately, to derive a pipeline instance that leads to high scores. 

\begin{figure}[t]
    \begin{center}
	    \includegraphics[width=0.98\columnwidth]{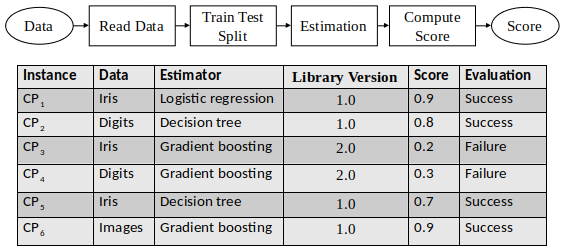}
	\end{center}
\vspace{-.3cm}
	\caption{Machine learning pipeline and its provenance. A data scientist can explore different input datasets and classifier estimators to identify a suitable solution for a classification problem.}
	\label{fig:pipeline}
\vspace{-.6cm}
\end{figure}

Analyzing the provenance of the runs, we can see that \emph{gradient boosting} leads to low scores for two of the datasets (\emph{Iris} and \emph{Digits}), but it has a high score for \emph{Images}. By contrast, \emph{decision trees} work well for both the \emph{Iris} and \emph{Digits} datasets, and \emph{logistic regression} leads to a high score for \emph{Iris}.

This may suggest that there is a problem with the \emph{gradient boosting} module for some parameters, that \emph{decision trees} provide a suitable compromise for different data, and that \emph{logistic regression} is good for the \emph{Iris} data. Because each run used different parameters for each method depending on the dataset, a definitive conclusion has to await additional testing of these hyperparameters.  Doing so manually is time-consuming and error-prone, while \ourapproach  automates this process.

\paragraph{Identifying Root Causes of Failures: Challenges.}
As the above examples illustrate, there are many potential causes for a given problem. %
Prior work used provenance to explain errors in computational processes that derive data~\cite{Wang:2015:DXD:2723372.2750549,GebalyFGKS14}. However, to test these hypotheses and obtain complete (and accurate) explanations, new pipeline instances must be executed that vary the different components of the pipeline.  

Trying all possible combinations of parameter-values leads to a combinatorial explosion of instances to execute,  and therefore can be prohibitively expensive. %
Thus, a critical challenge lies in the design of a  strategy that is provably efficient (often requiring only a linear number of pipeline executions in the number of parameters) for finding root causes. %
Causes of errors can include multiple parameters, each of which may have large domains.
So, it is important to have clear and concise explanations in terms of the parameter values already tried.

\paragraph{Contributions.} In this paper, we introduce \ourapproach, a new approach that makes  use of iteration and provenance to infer the root causes automatically and derive succinct explanations of failures in pipelines. Our contributions can be summarized as follows:
\begin{enumerate}
    \item 
    \ourapproach finds root causes autonomously and iteratively, intelligently selecting so-far untested combinations. %

    \item 
    We propose debugging algorithms that find root causes using fewer pipeline instances than state-of-the-art methods, avoiding unnecessary costly computations. In fact, \ourapproach often finds root causes using only a  number of pipeline instances linear in the number of parameters. 
    
  \item     The \ourapproach system further reduces time by exploiting parallelism, and
     \item Finally, \ourapproach derives concise explanations, to facilitate the tasks of human debuggers. %
\end{enumerate}

\paragraph{Outline}. The remainder of this paper is organized as follows. We review related work in Section~\ref{sec:relatedWork}. Section~\ref{sec:motivation} introduces the model we use for computational pipelines and formally defines the problem we address. In Section~\ref{sec:strategy}, we present algorithms to search for simple and complex causes of failures. We compare \ourapproach with the state of the art in Section~\ref{sec:experiments}
and conclude in Section~\ref{sec:conclusion}, where we outline directions for future work.

\section{Related Work} \label{sec:relatedWork}
\paragraph{Debugging Data and Pipelines.}
Recently, the problem of explaining query results and interesting features in data
has received substantial attention in the literature~\cite{Wang:2015:DXD:2723372.2750549, Bailis:2017:MPA:3035918.3035928, Chirigati:2016:DPM:2882903.2915245,DBLP:journals/pvldb/MeliouRS14, GebalyFGKS14}. 
Some have focused on explaining where and how errors occur in the data generation process~\cite{Wang:2015:DXD:2723372.2750549} and which data items are most likely to be causes of relational query outputs~\cite{DBLP:journals/pvldb/MeliouRS14,Wang:2017:QDE:3035918.3035925}.  Others have attempted to use data to explain \emph{salient} features in data (e.g., outliers) by discovering relationships among attribute values~\cite{Bailis:2017:MPA:3035918.3035928, Chirigati:2016:DPM:2882903.2915245, GebalyFGKS14}. 
In contrast,
\ourapproach aims to diagnose abnormal behavior in
computational pipelines that may be due to errors in  data, programs, or sequencing of operations. %

Previous work on pipeline debugging has focused on analyzing execution histories to identify problematic parameter settings or inputs, but such work does not iteratively infer and test new workflow instances.  %
Bala and Chana~\cite{Bala:2015:IFP:2775763.2776365} applied several machine learning algorithms to predict whether a particular pipeline instance will fail to execute in a cloud environment. The goal is to reduce the consumption of expensive resources by recommending against executing the instance if it has a high probability of failure. The system does not attempt to find the root causes of such failures.
Chen et al.~\cite{Chen2017} developed a system that identifies problems by finding the differences between  provenance (encoded as trees) of good and bad runs.
However, in general, these differences do not necessarily identify root causes, though they often contain them. 

Some systems have been developed to debug specific applications.
Viska~\cite{Gudmundsdottir2017} helps 
users identify the underlying causes for performance differences for a set of configurations. Users infer hypotheses by exploring performance data and then test these hypotheses by asking questions about the causal relationships between a set of selected features and the resulting performance.  Thus, Viska can be used to validate hypotheses but not identify root causes. %
Molly~\cite{DBLP:conf/sigmod/AlvaroRH15} combines the analysis of lineage with SAT solvers to find bugs in fault-tolerance protocols for distributed systems. It simulates failures, such as permanent crash failures, message loss, and temporary network partitions, in order to test fault-tolerance protocols over a specified period.

Although not designed for computational pipelines, Data X-Ray~\cite{Wang:2015:DXD:2723372.2750549} provides a mechanism for explaining the systematic causes of errors in the data generation process. The system finds shared features among corrupt data elements and produces a diagnosis of the problems.
Given the provenance of pipeline instances together with error annotations, Data X-Ray derives explanations consisting of features that describe the parameter-value pairs responsible for the errors. 
Explanation Tables~\cite{GebalyFGKS14} 
provides explanations for binary outcomes. Like Data X-Ray, it forms hypotheses based on a log of executions, but it does not propose new ones. %
Based on a table with a set of categorical columns (attributes) and one binary column (outcome), the algorithm produces interpretable explanations of the causes for the outcome in terms of the attribute-value pairs combinations. The explanations consist of a disjunction of patterns, and each pattern is a conjunction of attribute-value pairs. %
As discussed in Section~\ref{sec:experiments}, \ourapproach produces explanations that are similar to those of Data X-Ray and Explanation Tables, but they are also minimal and able to express inequalities and negations. Furthermore, \ourapproach employs a systematic method to intelligently generate new instances that enable it to derive concise explanations that are root causes for a problem.

\paragraph{Hyperparameter Tuning}
Our work is related algorithmically to approaches from hyperparameter tuning ~\cite{Bergstra2011, Bergstra:2013:MSM:3042817.3042832,Snoek:2012:PBO:2999325.2999464,Snoek:2015:SBO:3045118.3045349,Dolatnia2016},
since we can view the generation of new pipeline instances for debugging as an exploration of the space of its hyperparameters. %
Bayesian optimization methods are considered state of the art for the hyperparameter optimization problem~\cite{Bergstra:2012:RSH:2188385.2188395,Bergstra:2013:MSM:3042817.3042832,Snoek:2012:PBO:2999325.2999464,Snoek:2015:SBO:3045118.3045349,Dolatnia2016}. 
These methods approximate a probability model of the performance outcome given a parameter configuration that is updated from a history of executions. Gaussian Processes and Tree-structured Parzen Estimator are examples of probability models~\cite{Bergstra2011} used to optimize an unknown loss function using the \emph{expected improvement} criterion as acquisition function. To do this, they assume the search space is smooth and differentiable. This assumption, however, does not hold in general for arbitrary computational pipelines. Moreover, our goal is not to identify bad configurations (we usually have those, to begin with), but to identify the root cause(s), which are due to a subset of the parameters. Optimization, by contrast, seeks entire (in their case, good) configurations.

Examples of hyperparameter tuning techniques include OtterTune and BOAT. OtterTune~\cite{VanAken:2017:ADM:3035918.3064029} is a system that uses supervised learning techniques to find optimal settings of database system administrator knobs given a database workload and a set of metrics (optimization functions). 
BOAT~\cite{Dalibard:2017:BBA:3038912.3052662}
also optimizes database system configurations using Bayesian Optimization. However, instead of starting the optimization with a standard Gaussian process, it 
allows a user to input an initial probabilistic model that exploits previous knowledge of the problem.

\paragraph{Software Testing.}
State-of-the-art techniques for software testing~\cite{Johnson2018,Galhotra2017}, statistical debugging~\cite{Zheng2006,Liblit2005},  and bug localization~\cite{Gulzar2018,Attariyan2011,Attariyan2012} are often application-specific and/or require a user-defined test suite.
Some approaches require the instrumentation of binaries or source code in the form of predicates that can be observed during computational runs~\cite{Zheng2006,Liblit2005}. Such information, if available, can be helpful to localize and explain bugs. \ourapproach, however, does not assume any knowledge of the internal code of the computational processes:  it was designed to debug black-box pipelines where we can observe only the inputs and outputs. Hence, our explanations are expressed in terms of input parameters.
However, an interesting direction for future work would  be to consider
variables (or predicates) that can be observed but not manipulated in
our formalism to generate potentially richer explanations.
Approaches have also been proposed for bug localization in a  black-box scenario; however these were designed for specific applications and environments, e.g., Pinpoint for J2EE~\cite{chen2002pinpoint}. By contrast, \ourapproach was designed to support language-independent workflows.%

Automated test generation techniques also derive new tests (or instances in our terminology).  However, they do not aim to identify root causes (see, e.g.,~\cite{fraser@tse2013,godefroid@ndss2008,holler@uss2012}). One exception is 
Causal Testing~\cite{Johnson2018}. 
Similar to \ourapproach, Causal Testing aims to help users identify root causes for problems. However,  it requires the user to specify a (single) suspect variable to be investigated in a white-box scenario, while \ourapproach searches for potential causes for failures in a black-box scenario. Further  these causes may include multiple variables and value assignments.

\ourapproach helps a user to trace back the potential cause of a given behavior
to a component of a pipeline. Nevertheless, since a pipeline can
orchestrate a multitude of sophisticated tools, to identify and correct the
bug, it may be necessary to drill down into an individual component. If source code is available for that,  traditional debugging techniques can be used.

\paragraph{Identifying Denial Constraints.} 
Our approach is also related to the discovery of denial constraints in relational tables~\cite{BleifuB2017,Chu:2013:DDC:2536258.2536262}, particularly  functional dependencies. The  similarity can be illustrated as follows: imagine that there is a column indicating ``successful instance" or ``failed instance" for some set of parameter-values. Call it \emph{Success Or Fail}. If a failure occurs exactly when  parameter \texttt{A = 5} and \texttt{B = 6}, then that would manifest as a functional dependency  \texttt{AB $\xrightarrow{}$} \emph{Success Or Fail}, i.e., the result is a function of parameters \texttt{A} and \texttt{B}. However, if the failure happens when a disjunction holds, e.g., \texttt{A = 5} or \texttt{B = 6}, the same functional dependency would be inferred.  No more minimal functional dependencies such as \texttt{A $\xrightarrow{}$} \emph{Success Or Fail} would be inferred, because, for example, when \texttt{A = 4}, there can be success or failure depending on the \texttt{B} value.  Thus, functional dependencies are not expressive enough to characterize root causes.

\section{Definitions and Problem Statement}\label{sec:motivation}
Intuitively, given a set of computational pipeline instances, some of which lead to bad or questionable results, our goal is to find the root causes of failures, possibly by creating and executing new pipeline instances.

\begin{definition}{\textsc{(Pipeline, instance, parameter-value pairs,  value universe, results)}}
A \textbf{computational pipeline} (or workflow) $CP$ is a collection of programs connected together that contains a set of manipulable parameters $P$ (i.e., including hyperparameters, input data, versions of programs, computational modules). We denote as $CP_i$ a \textbf{pipeline instance} of $CP$ that defines values for the parameters for a particular run of $CP$. Thus, an instance $CP_i$ is associated with a list of \textbf{parameter-value pairs} $Pv_i$ containing an assignment $(p,v)$ for each $p \in P$. We denote by $CP_i[p]=v$ the assignment of value $v$ for parameter $p$ in the instance $CP_i$. For each parameter $p \in P$, the \textbf{parameter-value universe} $U_p$ is the set  of all property-values assigned to $p$ by any pipeline instance thus far, i.e., $U_p = \{ v | \exists i (p,v) \in CP_i \}$. The \textbf{Universe} $U = \{ (p, U_p ) | p \in P \}$.
\end{definition}

As we discuss in Section~\ref{sec:strategy}, the initial parameter-value universe $U$ can be expanded  by explicitly defining the parameter domains (e.g., parameter satisfaction can take integer values between 1 and 10).

\begin{definition}{\textsc{(Evaluation)}}
Let $E$ be a procedure that \textbf{evaluates} the result of an instance such that $E(CP_i)=\succeed$ if the results are acceptable, and $E(CP_i)=\fail$ otherwise. Normally, the evaluation procedure will be code that looks at some property of the result of a given pipeline instance.
\end{definition}

Thus a bug, for the purposes of this paper, is a collection of pipelines that, when executed, evaluate to $\fail$. Note that this is a deterministic definition that doesn't capture intermittent failures, e.g., timing bugs or non-deterministic failures. Even in such cases, however, if the bugs occur often enough, then \ourapproach may help, though without guarantee.

\begin{definition}{\textsc{(Hypothetical root cause of failure)}}
  Given a set of instances $G = $ $ CP_1, ... , CP_k$ and associated evaluations $E(CP_1 ), .... , E(CP_k )$, a \textbf{hypothetical root cause of failure} is a set %
  $C_f$ consisting of a Boolean conjunction of parameter-comparator-value triples (e.g., a triple may be of the form $A > 5$) which obey the following conditions among the instances $G$: (i) there is at least one $CP_i$ such that $Pv_i$ satisfies $C_f$ and $E(CP_i ) = \fail$; and (ii) if $E(CP_i ) = \succeed$, then the parameter-values pairs $Pv_i$ of $CP_i$ do not satisfy the conjunction $C_f$. 
  
\end{definition}

\paragraphemph{Example.} To illustrate the converse of point (ii), if $C_f$ = $A > 5$ and $B =7$, and $CP_i$ has the parameter values $A=15$ and $B=7$ and succeeds, then $C_f$ does not obey condition (ii) of a hypothetical root cause of failure.

$C_f$ is called \emph{hypothetical} because, based on the evidence so far, $C_f$ leads to $\fail$, but further evidence may refute that hypothesis. %

We should note that the root causes defined here should not be interpreted as the \emph{actual causes} of pipeline problems as characterized by causality theory~\cite{Pearl2009}. The goal of \ourapproach is to help the user identify sets of parameter-value pairs for which a black-box pipeline will always \fail. However, the root causes we output are not counterfactuals~\cite{lewis2013counterfactuals}, i.e., the pipeline would not necessarily \succeed had the root cause not been observed, because perhaps another root cause may come into play.  %
We simply want to determine the following implication definitively: $\textit{root-cause} \implies \fail$ for a single root cause.
\ourapproach can, however, also discover disjunctive combinations of configurations that lead to failure.

\begin{definition}{\textsc{(Definitive root cause of failure)}}
A \emph{hypothetical root cause of failure} $D$ is a 
\textbf{definitive root cause of failure} if there is no  instance $CP_q$ from the universe of $U$  with the property that $E(CP_q ) = \succeed$ and $Pv_q$ satisfies $D$. Informally, no pipeline instance that includes $D$ as a subset of its parameter-value settings leads to $\succeed$.
\end{definition}

\begin{definition}{\textsc{(Minimal Definitive Root Cause)}}
A definitive root cause $D$ is minimal if no proper subset of $D$ is a definitive root cause.
\end{definition}

The example in Figure~\ref{fig:pipeline} illustrates these concepts using the simple machine learning pipeline from the introduction. %
A possible evaluation procedure would test whether the resulting score is greater than 0.6. In this case, 
\texttt{Data} being different from \emph{Images} and \texttt{Estimator}  equal to \emph{gradient boosting} is a hypothetical root cause of failure. Section~\ref{sec:strategy} presents algorithms that  determine whether this root cause is definitive and minimal.

\paragraph{Problem Definition.} Given a computational pipeline $CP$ (e.g., a query, script, simulation) and a set of parameter-value pairs associated with previously-run instances $G = CP_1, ... , CP_k$,  we consider two goals: 
(i) to find at least one minimal definitive root cause or (ii) to find all minimal definitive root causes. Our cost measure for both goals is the number of executed pipeline instances beyond any given, previously run, instances.

\section{Debugging Algorithms}\label{sec:strategy}
Given a set of pipeline instances, \ourapproach identifies minimal definitive root
causes for failures.  As noted above, a naive strategy would be to try every
possible parameter-value pair combination of the parameter-value
universe, requiring the testing of a number of pipeline instances that is exponential in the number of
parameters. Instead, \ourapproach uses heuristics that turn out to be
quite effective at finding promising configurations. %

\ourapproach uses two iterative debugging algorithms in turn. The
first, called \shortcut, discovers definitive root causes (which we sometimes abbreviate to,
simply, bugs) consisting of a single conjunction of parameter-value (formally,  
parameter-equality-value) pairs.
The second, called \debuggingdecisiontrees and introduced in ~\cite{Lourenco2019}, discovers more complex definitive root causes
involving  inequalities (e.g., $A$ takes a value between $5$ and $13$). %

Because the results of the \debuggingdecisiontrees algorithm consist of disjunctions of conjunctions, they may contain redundancies, which we simplify using the Quine-McCluskey algorithm~\cite{DBLP:journals/corr/Huang14c}. The goal is to create concise explanations, making it easy for users to understand and act on them.

\subsection{Looking for Single Root Causes: The Shortcut Algorithm}\label{sec:minimal}

The \shortcut algorithm, shown in Algorithm~\ref{algo:shortcut},  starts from a pipeline instance $CP_f$ that
evaluates to \fail. It then uses pipeline instances that succeeded and are \emph{disjoint} , i.e., they share no parameter-values, from
$CP_f$  to construct new tests.

\begin{definition}[Disjoint Instances]
Two pipeline instances $CP_x$ and $CP_y$ are disjoint if $CP_x[p]\neq CP_y[p]$,$\forall p \in P$ associated to $CP$.   
\end{definition}

Intuitively, the \shortcut algorithm starts with the failing pipeline
instance $CP_f$ and a disjoint successful instance $CP_g$.
The existence of such a disjoint succeeding pipeline instance is a requirement for the theoretical results that follow and is called the {\em Disjointness Condition}. If the Disjointness Condition does not hold, then this method may still be useful as a heuristic.
 
The \emph{current} instance $CP_{current}$ is initialized to $CP_f$.
Then, using some order among parameters, for each parameter $p$, an
instance
$$CP_{current'}$$ is executed that consists of a copy of $CP_{current}$ except that $CP_{current'}[p] = CP_g[p]$.  If the instance $CP_{current'}$ fails then $CP_{current}$ is changed to $CP_{current'}$ and the next parameter is considered. The intution is that the value of $p$ in $CP_f$ did not cause the failure. In the end, the definitive minimal root
cause asserted by the \shortcut will be a subset of the pipeline
instance $CP_f$ that is still present in the final instance of
$CP_{current}$. We denote that subset as $D$.

The algorithm then performs a sanity check to see whether any superset
of the hypothetical minimal root cause $D$ is in an already executed
successful execution. If so, then the \shortcut algorithm has found a
proper subset of the definitive minimal root cause, but not an actual
definitive minimal root cause.
 
As noted above, if the Disjointness Condition does not hold, then the
\shortcut algorithm can still be used as a heuristic: take an instance
that differs in as many parameter-values as possible. While the
theoretical results that follow will not hold, this will often be good
enough, as the experimental results show (Section~\ref{sec:experiments}). %

Here is an example that illustrates how the \shortcut algorithm
works. %

\begin{example}\label{exemp:pipeline}
Consider the machine learning pipeline in Figure~\ref{fig:pipeline} again.
Here, the user is interested in investigating pipelines that lead to low F-measure scores and defines an evaluation function that returns \succeed if $\textit{score} \ge 0.6$ and \fail otherwise.  

For this pipeline, the user investigates three parameters:
\texttt{Dataset}, the input data to be classified; \texttt{Estimator},
the classification algorithm to be executed; and \texttt{Library
  Version} indicates the version of the machine learning library used. 
Table~\ref{tab:traces} shows examples of three executions of the pipeline.
\end{example}

\begin{table}[ht]
\centering
\caption{An initial (given) set of classification pipelines instances}
\label{tab:traces}
\resizebox{\columnwidth}{!}{
\begin{tabular}{|c|c|c|c|c|}
\hline
\textbf{Dataset} & \textbf{Estimator} & \textbf{Library Version} & \textbf{Score} & \textbf{Evaluation} ($\textit{score} \ge 0.6$)  \\ \hline
Iris&  Logistic Regression&  1.0&  0.9& \succeed            \\
\hline
Digits& Decision Tree&  1.0&  0.8& \succeed            \\
\hline
Iris& Gradient Boosting&  2.0&  0.2& \fail            \\
\hline
\end{tabular}
}
\end{table}

\begin{table}
\centering
\caption{Set of classification pipelines instances including the new
  instances (shown in blue) created by \shortcut by substituting values of parameters in $CP_f$ by corresponding values in $CP_g$.}
\label{tab:new_traces}
\resizebox{\columnwidth}{!}{
\begin{tabular}{|c|c|c|c|c|}
\hline
\textbf{Dataset} & \textbf{Estimator} & \textbf{Library Version} & \textbf{Score} & \textbf{Evaluation} ($\textit{score} \ge 0.6$)  \\ \hline
Iris&  Logistic Regression&  1.0&  0.9& \succeed            \\
\hline
Digits& Decision Tree&  1.0&  0.8& \succeed            \\
\hline
Iris& Gradient Boosting&  2.0&  0.2& \fail            \\
\hline
\textcolor{blue}{Digits}& \textcolor{blue}{Gradient Boosting}&  \textcolor{blue}{2.0}&  \textcolor{blue}{0.2}& \textcolor{blue}{\fail} \\
\hline
\textcolor{blue}{Digits}& \textcolor{blue}{Decision Tree}&  \textcolor{blue}{2.0}&  \textcolor{blue}{0.3}& \textcolor{blue}{\fail}           \\
\hline
\textcolor{blue}{\emph{Digits}}& \textcolor{blue}{\emph{Decision Tree}}&  \textcolor{blue}{\emph{1.0}}&  \textcolor{blue}{\emph{0.8}}& \textcolor{blue}{\emph{\succeed}}\\
\hline
\end{tabular}
}
\end{table}

In the initial traces shown in Table~\ref{tab:traces}, there are only
two disjoint instances with different evaluations:

\begin{description}
\item[] $CP_g$ = \{(Dataset,Digits), \\
        (Estimator,Decision Tree),\\
        (LibraryVersion,1.0)\} 
\item[] $CP_f$ =\{(Dataset,Iris), \\
(Estimator,Gradient Boosting), \\
(LibraryVersion,2.0) \}
\end{description}

Examining parameter \texttt{Dataset}, we replace its corresponding value in the current instance to be executed from Iris to Digits. Because the execution evaluates to \fail,  we keep this replacement in the current instance. Similarly, when we update the value of parameter \texttt{Estimator} to Decision Tree, the instance evaluation is still \fail, so we keep that replacement as well.

However, when \texttt{Library Version} is changed to $1.0$, the
resulting configuration evaluates to \succeed. This suggests that
\texttt{Library Version} $2.0$ may be the source of the problem.
Table~\ref{tab:new_traces}, displays all pipeline instances evaluated, including the new instances generated by the \shortcut algorithm.

For Pipelines with root causes similar to the ones in Example~\ref{exemp:pipeline}, the algorithm will find a minimal definitive root cause.

\RestyleAlgo{boxruled}
\begin{algorithm}
\SetAlgoLined
\KwIn{$CPI$, the set of pipeline instances in the execution history characterized by their parameter-values}
\KwIn{$E$, the evaluation function}
\KwIn{$P$, list of parameters}
\KwIn{$CP_f$, pipeline instance evaluated as \fail}
\KwIn{$CP_g$, pipeline instance evaluated as \succeed disjoint to $CP_f$}
\KwOut{$D$, asserted minimal definitive root cause}

\tcc{Initialization}	
$CP_{\text{current}} \gets CP_f$\;

\For{$p \in P$}{
    $CP_{\text{current}'} \gets CP_{\text{current}}$\;
    $CP_{\text{current}'}[p] \gets CP_g[p]$\;
    \If{$E(CP_{\text{current}'}) = \fail$}{
       $CP_{\text{current}} \gets CP_{\text{current}'}$\;
    }
}
$D  \gets CP_{\text{current}} \cap CP_f$\;

\For{$CP_i \in CPI$}{
    \If{$D \subseteq CP_i$ and $E(CP_i) = \succeed$ }{
       \KwRet{$\emptyset$}
    }
}

\KwRet{$D$}

 \caption{\shortcut Algorithm}\label{algo:shortcut}
\end{algorithm}

\begin{theorem}
If all definitive root causes are singleton para\-meter-values and the disjointness condition holds, then the shortcut algorithm will always assert exactly a minimal definitive root cause.
\label{theorem:singleton}
\end{theorem}

\vspace{-.2cm}

\begin{proof}
By construction. If all definitive root causes are singletons, then $CP_g$ cannot contain
any element of a root cause, otherwise $E(CP_g)=\fail$.  By contrast, $CP_f$ must contain at least one root cause. When iterating over parameter $p$, the \shortcut algorithm will replace $CP_f[p]$ by $CP_g[p]$ (because the values must be different on all parameters $p$ by the Disjointness Condition)  while there is still one root cause in $CP_{\text{current}}$. Therefore, by the end of the algorithm, only the the root cause would remain.
\vspace{-.2cm}        
\end{proof}

\paragraph{Guarantees of the Shortcut Algorithm.}
The \shortcut algorithm may be too aggressive in the sense that it can
return a root cause $D$ that is a proper subset of an actual minimal
definitive root cause of failure.

\vspace{-.1cm}

\begin{example}\label{exp:truncated}
Suppose that we have two minimal definitive root causes:
\begin{enumerate}
    \item $D_1 = \{(p_1,v_1),(p_2,v_2)\}$
    \item $D_2 = \{(p_1,v_1'),(p_3,v_3)\}$
\end{enumerate}
Consider also a computational pipeline consisting of three parameters $P=\{p_1,p_2,p_3\}$, and $CP_f$ and $CP_g$ as follows:
\begin{itemize}
    \item $CP_f = \{(p_1,v_1),(p_2,v_2),(p_3,v_3)\}$
    \item $CP_g = \{(p_1,v_1'),(p_2,v_2'),(p_3,v_3')\}$
\end{itemize}
Clearly $D_1 \subseteq CP_f$, therefore it is the root cause of the failure of $CP_f$. However, when iterating over parameter $p_1$, the \shortcut algorithm updates $CP_{\text{current}}[p_1]=v_1'$. But $E(CP_{\text{current}'})=\fail$ because $D_2 \subseteq CP_{\text{current}'}$. The same is observed when the algorithm iterates over parameter $p_2$. 
Consequently, the algorithm outputs $D=\{(p_3,v_3)\}$ as the root cause, but that is a proper subset of the minimal definitive root cause $D_2$.

\end{example}

In this case, we say that $D$ is a \textbf{truncated assertion}, i.e., it is too short. Note, however, $D$ will never be too long.

\begin{theorem}
The \shortcut algorithm never asserts a superset of a minimal
definitive root cause, provided the Disjointness Condition holds.
\label{theorem:minimality}
\end{theorem}

\begin{proof}
By contradiction. We assume that $\exists (p,v) \in D$, such that $(p,v)$ is not a necessary condition for an instance to \fail. By the construction at the beginning of the shortcut algorithm, if $(p,v) \in D$, $CP_f[p]=v$ and $CP_g[p]\neq v$ by the Disjointness Condition.

When the \shortcut algorithm iterates over parameter $p$, we observe $CP_{\text{current}}[p]=CP_f[p]$ and $CP_{\text{current'}}[p]=CP_g[p]$. 
Hence, since $(p,v)$ is not needed for an instance to \fail, at this
iteration, $E(CP_{\text{current'}})=\fail$, so $(p,v)$ would be
removed from {\text{current}} and therefore would never be asserted to
be part of the root cause. Contradiction.  
\end{proof}

To address the problem of truncated assertions, let us first observe another case when they do not arise, beyond the singleton case of Theorem~\ref{theorem:singleton}. 

\begin{example}\label{exp:sufficient}
Consider a slight modification of Example~\ref{exp:truncated}, where we add another parameter-value pair to $D_2$, defining the following scenario:
\begin{itemize}
    \item $D_1 = \{(p_1,v_1),(p_2,v_2)\}$
    \item $D_2 = \{(p_1,v_1'),(p_2,v_2''),(p_3,v_3)\}$
    \item $CP_f = \{(p_1,v_1),(p_2,v_2),(p_3,v_3)\}$
    \item $CP_g = \{(p_1,v_1'),(p_2,v_2'),(p_3,v_3')\}$
\end{itemize}
When iterating over parameter $p_1$, the \shortcut algorithm does not update $CP_{\text{current}}[p_1]=v_1$, since $E(CP_{\text{current}'})=\succeed$ because $D_1 \not\subseteq     CP_{\text{current}'}$ and $D_2 \not\subseteq CP_{\text{current}'}$. Similarly, the value of $CP_{\text{current}}[p_2]$ is not changed. Only $CP_{\text{current}}[p_3]$ is updated to $v_3'$.   
Thereafter, the algorithm would assert $D=\{(p_1,v_1),$ $(p_2,v_2)\}=D_1$ as minimal definitive root cause, which is correct.

\end{example}

In Example~\ref{exp:sufficient}, both $D_1$ and $D_2$ contain values for $p_1$ and $p_2$ that are distinct from their counterpart in the other definitive root cause, i.e., $D_1[p_1]\neq D_2[p_1]$ and $D_1[p_2]\neq D_2[p_2]$. We say that $D_1$ and $D_2$ are {\em sufficiently different}. This characteristic  directly influences when the \shortcut algorithm will yield truncated assertions and is formally defined as follows.

\begin{definition}[Sufficiently different instances]
Two definitive root causes $D_x$ and $D_y$ 
are {\bf sufficiently different} if (i) they share at least two properties and (ii) for all properties they have in common they differ in their values. Formally,

(i)
$| P_{D_x} \cap P_{D_y} | \ge 2$; 

(ii) and 
$D_1[p] \neq D_2[p], \forall p \in P_{D_x} \cap P_{D_y}$.
\end{definition}

\begin{theorem}
If the Disjointness Condition holds and all minimal definitive root causes are pairwise sufficiently different, then
the shortcut algorithm will never produce a truncated assertion.
\label{theorem:sufficientlydifferent}
\end{theorem}

\begin{proof}
  By contradiction.  Suppose there are two sufficiently different minimal definitive root
  causes $D_x$ and $D_y$, such that $D_x \subseteq CP_f$,
  $CP_{\text{current}}$ is initialized to $CP_f$, and at some point
  the \shortcut algorithm creates an instance $CP_{\text{current}} $
  such that $ D_y \subseteq CP_{\text{current}}$. We will show that
  this cannot happen.

  Consider the first parameter $p \in P_{D_x}$, such that 
\[
 CP_{\text{current'}}[p]=CP_g[p]  \ and \ 
 E(CP_{\text{current'}})=\fail
\]
Now, $D_x  \not\subseteq CP_{\text{current'}}$ because $D_x$ and $D_y$ differ on at least two properties. In addition,  $ D_y  \not\subseteq CP_{\text{current'}}$, since $CP_{\text{current'}}[p] $ is taken from $P_{D_x}$. Therefore, $E(CP_{\text{current'}})=\succeed$ because of the pairwise sufficient difference condition. Therefore, $CP_{\text{current}}[p]$ will not change its value. Thus, $ D_y \subseteq CP_{\text{current}}$ will never occur. 
\end{proof}

\paragraph{Stacked Shortcut Algorithm.}
Clearly, we cannot be sure {\em a priori} that all definitive root causes are single parameter-value pairs or that the minimal definitive root causes are sufficiently different, either of which would  ensure that the \shortcut makes no truncated assertions. However, even if neither holds, we may be able to avoid truncated assertions by a specific reapplication of \shortcut.

To see how, we first observe that \shortcut makes truncated assertions only if all elements of a minimal root cause are contained in the union of $CP_f$ and $CP_g$. This {\em union property} is formally described in Theorem~\ref{thm:union}. 

\begin{theorem}\label{thm:union}
The shortcut algorithm will yield a truncated assertion 
for a given $CP_f$ and $CP_g$ only if there is a minimal definitive root cause $D$, such that $D\subseteq CP_f \cup CP_g$ and $D \not\subset CP_f$.
\end{theorem}

\begin{proof}
In the course of the \shortcut algorithm, all property values in $CP_{\text{current}}$ come from $CP_f$ or $CP_g$. By construction, the asserted root cause is the intersection of $CP_f$ and $CP_{\text{current}}$. So if the asserted root cause is truncated, $CP_{\text{current}}$ must have elements from $CP_g$ that cause $CP_{\text{current}}$ to evaluate to $\fail$. Therefore there is a minimal definitive root cause in the union of $CP_f$ and $CP_g$.
\end{proof}

Based on the previous theorems, we extended the shortcut algorithm to the \stacked algorithm which basically runs a given failed configuration $CP_f$ individually against multiple disjoint good configurations and then takes the union of the inferred root causes.  Algorithm~\ref{algo:stacked} shows the algorithm's pseudo-code. \stacked  is guaranteed to produce a correct solution if \ourapproach can find $k$ {\bf mutually disjoint} successful instances, and there are at most $k$ distinct minimal root causes.   

Recall that two instances $CP_1$ and $CP_2$ are  disjoint if they have different values for all properties. That is, $ \forall p CP_1[p] \neq CP_2[p]$. A set of instances is mutually disjoint if every pair of instances are disjoint.

\begin{algorithm}[ht]
\SetAlgoLined
\KwIn{$CPI$, the set of pipeline instances in the execution history characterized by their parameter-values}
\KwIn{$E$, the evaluation function}
\KwIn{$P$, list of parameters}
\KwOut{$D$, asserted minimal definitive root cause}

\tcc{Initialization}	
$D  \gets \emptyset$\;
\tcc{Find an instance that evaluates to \fail}
Let $CP_f$ be such that $CP_f \in CPI$, and $E(CP_f)=\fail$\;

\tcc{Find $k$ successful instances disjoint with respect to $CP_f$ and mutually disjoint if possible}
$CPG \gets \{CP_1, CP_2, ...,CP_k\}$, such that $CP_i$, for $i\in\{1,2,...k\}$, are mutually disjoint and $E(CP_i)=\succeed$\;

\For{$CP_g \in CPG$}{
    $D  \gets D \cup \textit{shortcut}(CPI,E,P,CP_f,CP_g)$\;    
}

\KwRet{$D$}
 \caption{\stacked Algorithm}\label{algo:stacked}
\end{algorithm}

\begin{theorem}
If all $CP_i$, such that $E(CP_i)=\succeed$, for $i \in \{1,2,...,k\}$, are mutually disjoint and disjoint from $CP_f$, and there are fewer than or equal to $k$ distinct minimal definitive root causes, then the \stacked Algorithm will never make a truncated assertion.
\end{theorem}

\begin{proof}
By construction. For each other minimal definitive root cause $D \not\subseteq CP_f$, there can be at most one $CP_i$ with the property that $D \subseteq CP_i \cap CP_f$, since all instances are disjoint. Because there are fewer than $k$ distinct minimal definitive root causes by assumption, there exists at least one $CP_i$, which does not have the union property with respect to $CP_f$.
So, by the construction of $D$, the \stacked algorithm will yield an assertion (candidate root cause) that is not truncated.%
\end{proof}

Note that even if all successful instances are not mutually disjoint (perhaps because some parameters have very few values), each additional call to \textit{shortcut} (i.e., each call to \shortcut with a different disjoint good instance) reduces the likelihood of yielding a truncated assertion. The reason is that the second-to-last line of the \stacked algorithm can only grow the hypothetical root causes.

Finally, note that both \shortcut and \stacked are linear in the number of parameters, a very useful property when there are hundreds of parameters having at least two values each.

\subsection{Finding Bugs with Inequalities: Debugging Decision Trees}\label{sec:decision}

While the \shortcut and \stacked algorithms can find a single minimal definitive
root cause very efficiently, usually without truncation (as we will
see in the experimental section), characterizing all minimal definitive root
causes is challenging. For this purpose, we use an algorithm
that is exponential (in the number of parameters) in the worst case,
but can characterize inequalities as well as equalities and does well
heuristically even with a small
budget~\cite{Lourenco2019}. 

The algorithm constructs a \emph{debugging decision tree} using the
parameters of the pipeline as features and the evaluation of the
instances as the target. Thus a leaf is either purely \succeed,
if all pipeline instances so far tested that lead to that leaf evaluate to \succeed; or
\fail, if all pipeline instances leading to that leaf evaluate to \fail,
or \emph{mixed}.  The algorithm works as follows:

\begin{enumerate}
\item Given an initial set of instances $CPI$, construct a decision tree based on the evaluation results for those instances (\succeed or \fail). An inner node of the decision tree is a triple (\textit{Parameter},\textit{Comparator},\textit{Value}), where the \textit{Comparator} indicates whether a given \textit{Parameter} has a value equal to, greater than (or equal to), less than (or equal to), or unequal to \textit{Value}.  

\item If a conjunction involving a set of parameters, say, $P_1$ $P_2$, and $P_3$, leads to a consistently failing execution (a pure leaf in decision tree terms), then that combination becomes a suspect.

\item Each suspect is used as a filter in a Cartesian product of the parameter values from which new experiments will be sampled.
\end{enumerate}

Step 3 requires some explanation. 
Consider an example where all comparators denote equality. Suppose a path in the decision tree consists of  $P_1=v_1$, $P_2=v_2$, and $P_3=v_3$. To test that path, all other parameters will be varied. If every instance having the parameter-values $P_1=v_1$, $P_2=v_2$, and $P_3 = v_3$ leads to failure, then  that conjunction constitutes a \emph{definitive root cause of failure}. 

If the path consists of non-equality comparators (e.g., $P_1=v_1$, $P_2=v_2$, and $P_3 > 6$), then the algorithm chooses a satisfying value for each of those parameters as a prototype, (e.g., $P_3=7$) and choose pipeline instances having those values (e.g., all pipelines $P_1=v_1$, $P_2=v_2$, and $P_3 = 7$).
Conversely, if any of the newly generated instances presents a  (\succeed) pipeline instance, the decision tree is rebuilt, taking into account the whole set of executed pipeline instances $CPI$, and a new suspect path is tried.

Note  that \ourapproach uses decision trees in an unusual way. We are
not trying to predict whether an untested configuration will lead to
\succeed or \fail, but simply use the tree to discover short paths,
possibly characterized by inequalities, that lead to \fail. Those will
be our suspects. For that reason, we build a complete decision tree,
i.e., with no pruning.

\subsection{Parallelism}
The most time-consuming aspect of debugging is the execution of pipeline instances. Fortunately,  each pipeline instance is independent.
Hence different instances can be run in parallel. However, such an approach may lead to the execution of pipelines that are ultimately unnecessary (e.g., if one pipeline instance shows that $A.v$ is not a definitive root cause, then further tests on $A.v$ may not be useful). If the search space is large, this extra overhead turns out to be small, as we show in Section~\ref{sec:scalability}.

\section{Experimental Evaluation}\label{sec:experiments}
To evaluate the effectiveness of \ourapproach, we compare it against state-of-the-art methods for deriving explanations as well as for hyperparameter optimization, using both real and synthetic pipelines. 
We examine different scenarios, including when a single minimal definitive root cause is sought and when a budget for the number of instances that can be run is set. 
We also evaluate the scalability of \ourapproach when multiple cores are available to execute pipeline instances in parallel, and when the number of parameters and values increase.

\paragraph{Baselines.} Because no previous approach both creates new instances and derives explanations, we compare our approach against combinations of state-of-the-art methods.
We use Data X-Ray~\cite{Wang:2015:DXD:2723372.2750549} and Explanation Tables~\cite{GebalyFGKS14} to derive explanations. To generate instances for all explanation algorithms, we use both the instances from \ourapproach and Sequential Model-Based Algorithm Configuration (SMAC)~\cite{HutHooLey11-smac}.

SMAC is a method for hyperparameter optimization that is often more effective
at searching configuration spaces than grid search~\cite{Bergstra:2012:RSH:2188385.2188395}. 
We also ran experiments using random search as an alternative, i.e., randomly generating instances and then analyzing them. However, the results were always worse than those obtained using SMAC or \ourapproach. Therefore, for simplicity of presentation and to avoid cluttering the plots, we omit the random search results.  

The explanation approaches analyze the provenance of the pipelines, i.e., the instances previously run and their results, but do not suggest new ones.  By contrast, SMAC iteratively proposes new pipeline instances, but it always outputs a complete pipeline instance: the best it can find given a budget of instances to run and a criterion. This procedure makes sense for SMAC's primary use case, which is to find a set of parameter-values that performs well, but it is less helpful for debugging because it does not attempt to find a minimal root cause. For example, if a minimal definitive root cause of a pipeline is that parameter $P_i$ must have a value of $5$, SMAC will return a pipeline that fails, which has $P_i$ set to $5$. But since the pipeline may have many other parameter-values, the user has no way of knowing that $P_i = 5$ is the minimal definitive root cause and thus gains no insight into how to rectify the bug.

To give the explanation methods a reasonable chance to find minimal root causes,  we combine the explanations with the generative techniques. We apply Data X-Ray and Explanation Tables to suggest root causes for the pipeline instances generated by SMAC, and also feed both methods with the instances created by \ourapproach. Since SMAC looks for good instances, mostly for machine learning pipelines, we change its goal to look for bad pipeline instances. 

\paragraph{Evaluation Criteria.} 
We consider two goals: (i) {\em FindOne} -- find at least one minimal definitive root cause in each pipeline; (ii) {\em FindAll} -- find all minimal definitive root causes. The use case for {\em FindOne} is a debugging setting where it might be useful to work on one bug at a time, in the hope that resolving one may resolve or at least mitigate others. The use case for {\em FindAll} is when a team of debuggers can work on many bugs in parallel. {\em FindAll} may also be useful to provide an overview of the set of issues encountered.
We use precision and recall to measure quality. These are defined differently for the {\em FindOne} case than for the {\em FindAll} case.

Formally, let $UCP$ be a set of computational pipelines, where each pipeline $CP \in UCP$ (e.g., the pipeline of Figure~\ref{fig:pipeline}) is associated with a set of minimal definitive root causes $R(CP)$. 
Given a set of root causes $A(CP)$ asserted by an algorithm $A$ on pipeline $CP$ for the {\em FindOne} case, we check if $A(CP)$ has at least one actual root cause. Precision is then the number of computational pipelines for which at least one minimal definitive root cause is found divided by the sum of the total number of pipelines where at least one minimal definitive root cause is found and the number of false positives (predicted root causes that are not, in fact, minimal definitive root causes). Formally, the \textit{precision for FindOne} is:
\begin{equation*}
   \frac{\sum_{CP \in UCP} |A(CP) \cap R(CP)\neq \emptyset| }{\sum_{CP \in UCP} |A(CP) \cap R(CP)\neq \emptyset| +  |A(CP) - R(CP)|}
\end{equation*}
\noindent where  $A(CP) \cap R(CP) \neq \emptyset$ evaluates to 1 if $A(CP)$ corresponds to at least one of the conjuncts in $R(CP)$. Recall for {\em FindOne} is the fraction of the $|UCP|$ pipelines for which a minimal definitive root cause is found by $A$. The \textit{recall for FindOne} is thus:
\begin{equation*}
   \frac{\sum_{CP \in UCP} |A(CP) \cap R(CP) \neq \emptyset| }{|UCP|} 
\end{equation*}
\noindent  For {\em FindAll}, precision is the fraction of root causes that $A$ identifies that are, in fact, minimal definitive root causes. The \textit{precision for FindAll} is defined as: 
\begin{equation*}
   \frac{\sum_{CP \in UCP} |A(CP) \cap R(CP)| }{\sum_{CP \in UCP} |A(CP)|}
\end{equation*}
\noindent \emph{Recall for FindAll}  is the fraction of all 
the $R(CP)$ minimal definitive root causes, for all $CP \in UCP$, that are found by the algorithms: 
\begin{equation*}
   \frac{\sum_{CP \in UCP} |A(CP) \cap R(CP)|}{\sum_{CP \in UCP} |R(CP)|} 
\end{equation*}
\noindent For both {\em FindOne} and {\em FindAll}, we also report the F-measure, i.e., the harmonic mean of their respective measures of precision and recall.
\begin{equation*}
   \textit{F-measure}=2\times\frac{\textit{Precision}\times\textit{Recall}}{\textit{Precision}+\textit{Recall}} 
\end{equation*}

Our first set of tests allows \ourapproach to find at least one minimal definitive root cause using each of its algorithms (\shortcut, \stacked, and \debuggingdecisiontrees). The experiment then grants the same number of instances to all other methods. Thus, the precision and recall for each algorithm is based on the same instance budget.%

In these tests, Data X-Ray and Explanation Tables are given (i) the instances  generated by \ourapproach and, in a separate test, (ii) the instances generated by SMAC.

\paragraph{Pipeline Benchmark.} We evaluate our approach using both synthetic and real pipelines. 
We have created synthetic data that reflect typical pipelines in data science and computational science, which often involve multiple components and associated parameters. The pipelines have between three and fifteen parameters, and each parameter has between five and thirty values.
The parameter values are either ordinal (e.g., temperature) or categorical (e.g., color), each with probability 1/2. Each synthetic pipeline consists of a parameter space and a definitive root cause of failure automatically generated as follows:    

\begin{enumerate}
    \item We uniformly sample a non-empty subset of parameters to be part of a conjunction.
    \item For each parameter in the subset, we uniformly sample from its values.
    \item For each parameter-value pair,  we uniformly sample from the set of comparators $C = \{ = , \le, >, \neq \}$. 
    \item After adding a conjunctive root cause, we add another conjunctive root cause with a certain probability. 
\end{enumerate}

The example below illustrates the parameter space and the definitive root cause for one of the synthetic pipelines.

\begin{example}
A pipeline having  three parameters with four possible values each could be characterized as follows: 

\begin{itemize}
    \item Parameter Space: $p_1 \in \left[1.0,2.0,3.0,4.0\right]$, $p_2 \in \left[1,2,3,4\right]$, and $p_3 \in \left[``p31",``p32",``p33",``p34"\right]$. 
    \item Minimal definitive Root Cause : $(p_1 = 4)$ or $(p_2 < 3.0$  and  $p_3 \neq ``p34")$.
\end{itemize}
\end{example}

We also evaluate the debugging strategies on real-world computational pipelines (see Section~\ref{sec:real}).

\paragraph{Implementation and Experimental Setup.}
The current prototype of \ourapproach %
contains a dispatching component that runs in a single thread and spawns multiple pipeline instances in parallel.
In our experiments, we used five execution engine workers to run the instances. 

We used the SMAC version for Python 3.6. We also used the code, implemented by the respective authors, for both the Data X-Ray algorithm (implemented in Java 7)~\cite{Wang:2015:DXD:2723372.2750549} and Explanation Tables~\cite{GebalyFGKS14} (written in python 2.7).
Since Data X-Ray does not generate new tests, we use the pipeline instances created by \ourapproach as input to the feature model input of Data X-Ray. Separately, we converted the pipeline instances created by SMAC as input to the feature model of Data X-Ray. 
Similarly, we used the pipeline instances generated by both \ourapproach and SMAC to populate the database schema required by Explanation Tables. 

All experiments were run on a Linux Desktop (Ubuntu 14.04, 32GB RAM, 3.5GHz $\times$ 8 processor). For purposes of reproducibility and community use, we made our code and experiments available (\url{https://github.com/ViDA-NYU/BugDoc}).

\subsection{Synthetic Pipelines}\label{sec:synthetic}

The results for the synthetically generated pipelines are reported according to the characteristics of their definitive root causes. The characteristics span three scenarios, consisting of multiple pipelines and covering different lengths of definitive root causes: 

\begin{enumerate}
\item a single parameter-comparator-value triple; 
\item a single conjunction of triples containing parameter-comparator-value; and
\item a disjunction of conjunctions of parameter-comparator-value triples.
\end{enumerate}

\noindent These scenarios are useful to assess the generality and expressiveness of the different approaches to explanation.

\begin{figure*}[t]
  \centering
  \includegraphics[width=\textwidth]{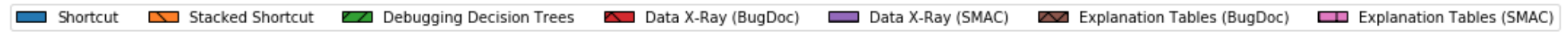}
  \centering
  \subcaptionbox{Precision \label{fig:single_one_precision}}{
    \includegraphics[width=0.67\columnwidth]{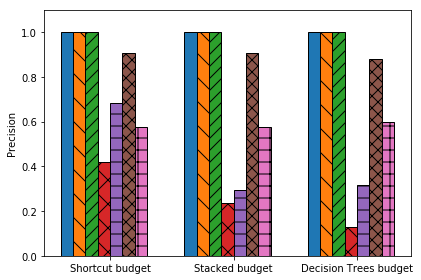}
  }
  \subcaptionbox{Recall
  \label{fig:single_one_recall}}{
    \includegraphics[width=0.67\columnwidth]{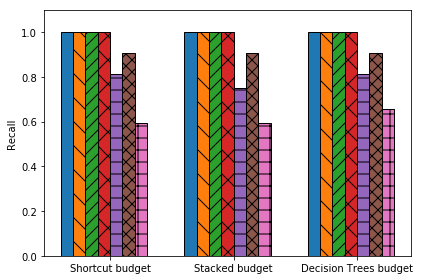}
  }
  \subcaptionbox{F-measure\label{fig:single_one_score}}{
    \includegraphics[width=0.67\columnwidth]{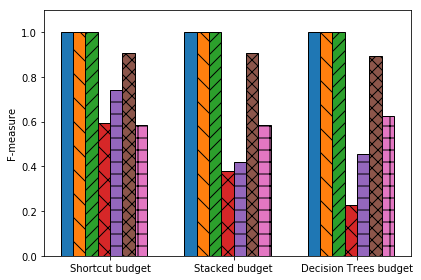}
  }
  \centering
  \subcaptionbox{Precision \label{fig:conjunction_one_precision}}{
    \includegraphics[width=0.67\columnwidth]{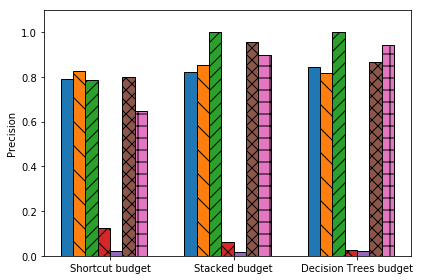}
  }
  \subcaptionbox{Recall
  \label{fig:conjunction_one_recall}}{
    \includegraphics[width=0.67\columnwidth]{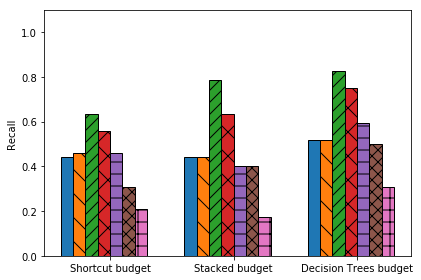}
  }
  \subcaptionbox{F-measure\label{fig:conjunction_one_score}}{
    \includegraphics[width=0.67\columnwidth]{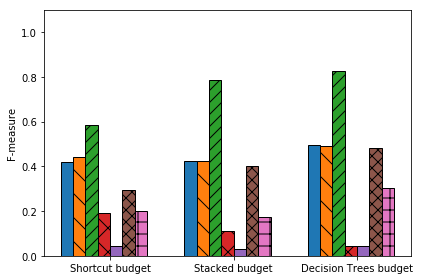}
  }
  \centering
  \subcaptionbox{Precision \label{fig:disjunction_one_precision}}{
    \includegraphics[width=0.67\columnwidth]{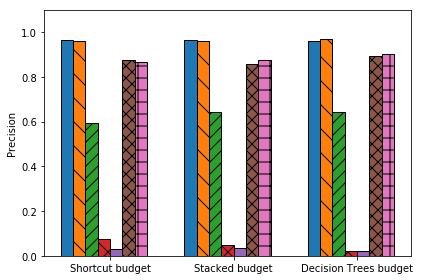}
  }
  \subcaptionbox{Recall
  \label{fig:disjunction_one_recall}}{
    \includegraphics[width=0.67\columnwidth]{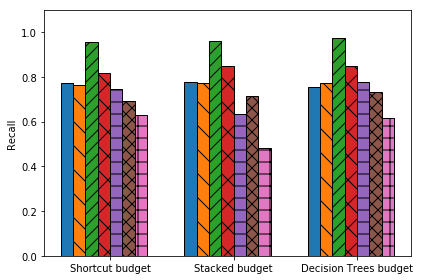}
  }
  \subcaptionbox{F-measure\label{fig:disjunction_one_score}}{
    \includegraphics[width=0.67\columnwidth]{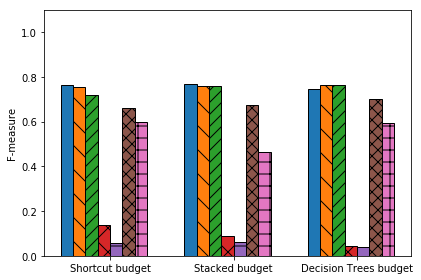}
  }
\vspace{-.3cm}
  \caption{Synthetic Pipelines. Metrics for the {\em FindOne} problem  when the root cause is a single parameter-value-comparator (top row, Figures~\ref{fig:single_one_precision},~\ref{fig:single_one_recall}, and~\ref{fig:single_one_score}), a single conjunction (middle row, Figures~\ref{fig:conjunction_one_precision},~\ref{fig:conjunction_one_recall}, and~\ref{fig:conjunction_one_score}), or a disjunction of conjunctions (bottom row, Figures~\ref{fig:disjunction_one_precision},~\ref{fig:disjunction_one_recall}, and~\ref{fig:disjunction_one_score}). In each figure, the leftmost group uses as many instances as does \shortcut, the middle uses as many as \stacked, the rightmost as many as \debuggingdecisiontrees. %
 }
  \label{fig:find_one}
\end{figure*}

\paragraph{Precision, Recall, and F-measure.} Figure~\ref{fig:find_one} shows the precision, recall, and F-measure for the  {\em FindOne} problem for 
the three types of definitive root causes. In the horizontal axis of each plot, we group all debugging methods by the maximum number of instances they were allowed to use to derive explanations, i.e., the number of new instances it took \shortcut, \stacked with four shortcuts, and \debuggingdecisiontrees to solve the problem.

\ourapproach's algorithms outperform Data X-Ray and Explanation Tables in all three scenarios, both when the baselines use instances generated by \ourapproach and SMAC.  If the definitive root cause is a single parameter-comparator-value (Figures~\ref{fig:single_one_precision},~\ref{fig:single_one_recall}, and~\ref{fig:single_one_score}), \shortcut and \stacked achieve similar precision and recall to \debuggingdecisiontrees, which dominates the other scenarios. %

Since we look for individual parameter-comparator-value triples with \shortcut and disjoint patterns in the data with decision trees, the likelihood that \shortcut does not find a definitive answer is higher in the scenario where a definitive root cause is a conjunction of factors, as can be seen in the relatively lower recall in Figure~\ref{fig:conjunction_one_recall}. Conjunctions that are composed of equalities and inequalities have a high probability of presenting configurations with the union property. Hence the \shortcut and \stacked algorithms generate more truncated assertions, and their precision score is lower in Figure~\ref{fig:conjunction_one_precision} as compared to Figures~\ref{fig:single_one_precision} and~\ref{fig:disjunction_one_precision}. However, the shortcut algorithms still give better performance than the state-of-the-art algorithms.

Also note that in most cases, the state-of-the-art methods using instances generated by \ourapproach outperform those methods using the SMAC instances. This suggests that our approach  effectively proposes more useful test cases.

\begin{figure*}[t]
  \centering
  \includegraphics[width=\textwidth]{imgs/legends.png}
  \centering
  \subcaptionbox{Precision \label{fig:disjunction_all_precision}}{
    \includegraphics[width=0.67\columnwidth]{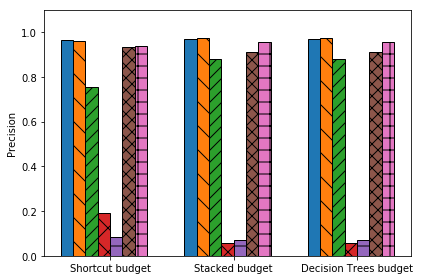}
  }
  \subcaptionbox{Recall
  \label{fig:disjunction_all_recall}}{
    \includegraphics[width=0.67\columnwidth]{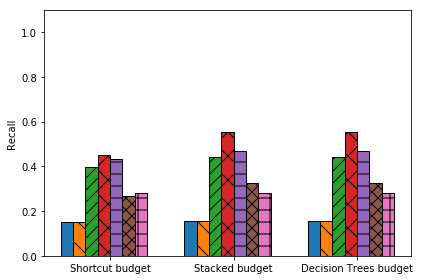}
  }
  \subcaptionbox{F-measure\label{fig:disjunction_all_score}}{
    \includegraphics[width=0.67\columnwidth]{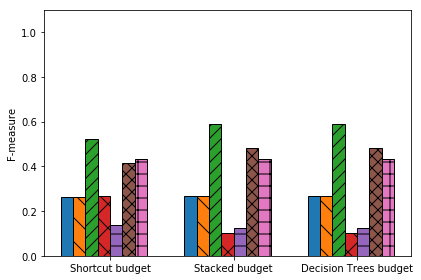}
  }
\vspace{-.3cm}
   \caption{Synthetic Pipelines. Metrics for the {\em FindAll} problem when the root cause is a disjunction of conjunctions (Figures~\ref{fig:disjunction_all_precision},~\ref{fig:disjunction_all_recall}, and~\ref{fig:disjunction_all_score}). In each sub-figure, the leftmost group uses as many instances as does \shortcut, the middle group as many \stacked, the rightmost as many as \debuggingdecisiontrees. %
 }
  \label{fig:find_all}
\end{figure*}

\begin{figure*}[ht]
  \centering
  \includegraphics[width=\textwidth]{imgs/legends.png}
  \centering
  \subcaptionbox{Parameters per asserted root cause \label{fig:avg_params}}{
    \includegraphics[width=0.8\columnwidth]{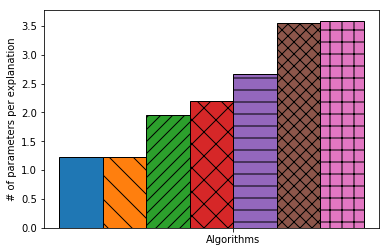}
  } \ \ \ \ 
  \subcaptionbox{Log of the number of asserted root causes per actual definitive root cause
  \label{fig:avg_answers}}{
    \includegraphics[width=0.8\columnwidth]{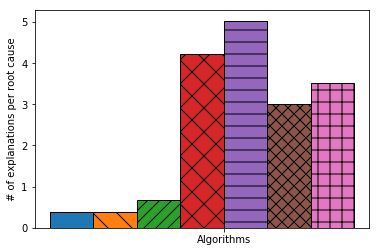}
  }
 \vspace{-.3cm}
   \caption{Synthetic Pipelines. (a) Average number of parameters per asserted root causes for each algorithm and (b) average logarithmic number of asserted root causes per actual definitive root cause for each method. 
 }
\label{fig:minimality}
\end{figure*}

Similar relative results hold for the {\em FindAll} problem Figure~\ref{fig:find_all} shows, although we observe the expected decrease in recall in Figure~\ref{fig:disjunction_all_recall}, as a single root cause is no longer sufficient. %
The non-minimal approach of Data X-Ray pays off in this scenario with multiple reasons for a pipeline to \fail.
However, \debuggingdecisiontrees presents a better trade-off between precision and recall (Figure~\ref{fig:disjunction_all_score}).  

\paragraph{Discussion.} 
The answers provided by Explanation Tables represent a prediction of the pipeline instance evaluation result expressed as a real number, were $1.0$ corresponds to a root cause. %
The precision of Explanation Tables is always high, but the recall is usually low. 
The converse happens with Data X-Ray, whose precision is low, but the recall is high. 
The reason for this is that Data X-Ray provides explanations that are not minimal definitive root causes. Further, neither Data X-Ray nor Explanation Tables support negation and inequality. 

Because both Data X-Ray and Explanation Tables achieved higher performance  when using the instances generated by \ourapproach than when using the instances generated by SMAC, we omit the SMAC configurations from the case studies with real-world pipelines presented later in this section. %

The takeaway message from the experiments is that \ourapproach dominates the other methods based on F-measure in every case, with \debuggingdecisiontrees dominating the shortcut methods unless the budget is small.

\paragraph{Conciseness of Explanation.} Figure~\ref{fig:minimality} shows that \ourapproach's algorithms not only provide explanations that are more concise in the number of parameters than Data X-Ray and Explanation Tables (Figure~\ref{fig:avg_params}) but also that it does not assert more root causes than there are (Figure~\ref{fig:avg_answers}). 

\subsection{Scalability}\label{sec:scalability}

\begin{figure}[b]
    \begin{center}
	    \includegraphics[width=0.87\columnwidth]{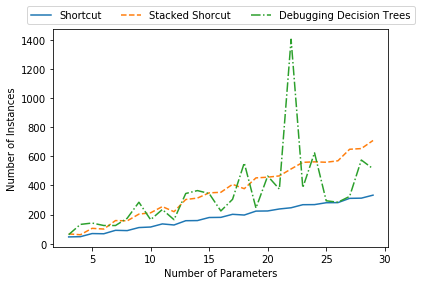}
	\end{center}
	\vspace{-.3cm}
	\caption[]{Instances required to execute each algorithm as a function of the number of parameters. %
	}
	\label{fig:scalingup}
	
\end{figure}

The primary computational cost for all algorithms we consider is the cost of running the pipeline instances. Figure~\ref{fig:scalingup} shows the number of instances created by each of \ourapproach's algorithms as a function of the number of parameters of the pipeline. \shortcut and \stacked increase linearly as expected. 
Because the time performance of \debuggingdecisiontrees has no simple relationship with root causes and could be exponential with the number of parameters, the user should choose  \shortcut or \stacked if there are many parameters and instances are expensive to run.

As noted above, the pipeline instances to test  can be run in parallel, but at some risk to unnecessary computation. To evaluate scalability, we re-execute the experiment with synthetic data, described in Section~\ref{sec:synthetic}, with different numbers of parallel computational cores and checked how many instances each core processed.
As Figure~\ref{fig:parallelism} shows, the scale-up is essentially linear with the number of cores for the \debuggingdecisiontrees algorithm solving the {\em FindAll} problem. Thus given sufficient computing power, even \debuggingdecisiontrees can explore relatively large parameter spaces.

\begin{figure}[t]
    \begin{center}
	    \includegraphics[width=0.87\columnwidth]{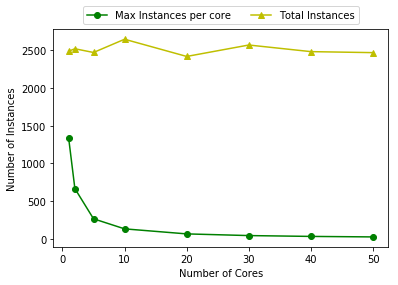}
	\end{center}
	\vspace{-.3cm}
	\caption[]{Scalability of \ourapproach when	running the \debuggingdecisiontrees algorithm on multiple cores. }
	\label{fig:parallelism}
	
\end{figure}

\subsection{Real-World Pipelines}\label{sec:real}

\paragraph{Data Polygamy Framework.}

Data Polygamy aims to discover statistically significant relationships in a large number of spatio-temporal datasets~\cite{Chirigati:2016:DPM:2882903.2915245}.
We created a VisTrails~\cite{Freire2011} pipeline that reproduces  an experiment designed by the Data Polygamy authors to evaluate the p-value and false discovery rate for their approach under different scenarios. Specifically, the pipeline evaluates different methods for determining statistical significance.
The datasets used are synthetically generated, and their features are given as input parameters for the experiment. This process is a good use case for our approach because it has the following properties:

\begin{myitemize}
    \item The experiment requires a complex pipeline, including steps for data cleaning, data transformation, feature identification, multiple hypotheses testing, and other activities.
    \item The input data is heterogeneous -- over 300 datasets at different spatio-temporal resolutions.
    \item The parameter space is large, consisting of 2 boolean, 3 categorical (3 to 10 possible values), and 7 numerical parameters. Each instance takes 20 minutes to run, making manual debugging impractical.
\end{myitemize}

For this experiment, we selected different data types and steps of the computational pipeline. %
Each parameter can conceivably take on any value belonging to its type (e.g., Integer or Boolean). Given a set of pipeline instances, some of which crash and some of which execute to completion, we want to find at least one minimal set of parameter-values or combinations of parameter-values among those in the given pipeline instances, which cause the execution to crash. %

\paragraph{GAN Training.}
Generative adversarial networks (GAN)~\cite{goodfellow2014generative} are widely applied to image generation and semi-supervised learning~\cite{radford2015unsupervised,Zhang_2017}. Training these generative models involves an expensive computational process with several configuration parameters, such as the architecture choices and a high-dimensional hyperparameter space to tune.  Sequence model-based approaches like Bayesian Optimization are prohibitively expensive in practice, since a single configuration could take more than a week to train. The most extensive study on the pathology of GAN training~\cite{brock2018large} entailed modifying baseline architectures and setting hyperparameters manually over three months, using hundreds of cores of a Google TPUv3 Pod~\cite{Jouppi:2017:IPA:3140659.3080246}. \citet{lucic2017gans} evaluated seven different GAN architectures and their hyperparameter configurations, performing a random search in an experimental setting that would take approximately 6.85 years using a single NVIDIA P100.

We created a computational pipeline that trains a modified SAGAN~\cite{zhang2018selfattention} on CIFAR-10~\cite{krizhevsky2009learning} and applied \ourapproach to find root causes of one of the most common problems of GAN training: \emph{mode collapse}~\cite{che2016mode}. Our evaluation function sets a threshold on the Frechet Inception Distance (FID)~\cite{heusel2017gans} metric, which is a proxy for mode collapse. This pipeline specified only 6 parameters limited to 5 possible values. The bottleneck was the execution time because each configuration is trained in approximately 10 hours, depending on the discriminator and generator learning rates and the number of steps.    

\paragraph{Transactional Database Performance.}
DBSherlock~\cite{Yoon:2016:DPD:2882903.2915218} is a tool designed to help database administrators diagnose online transaction processing (OLTP) performance problems. DBSherlock analyzes hundreds of statistics and configurations from OLTP logs and tries to identify which subsets of that data are potential root causes of the problems. In their experiments, the authors ran different settings of the TPC-C benchmark~\cite{TPCC}, introducing 10 distinct classes of performance anomalies varying the duration of the abnormal behavior.  For each type of anomaly, they collected the workload logs, creating a dataset of logs, each labeled as normal or anomalous.      

This dataset was used by~\citet{Bailis:2017:MPA:3035918.3035928} to demonstrate Macrobase's ability to distinguish abnormal behavior in OLTP servers, where a classifier was trained to identify servers presenting degradation in performance. 

We ran \ourapproach on this data to identify the root causes of each class of performance anomaly. This experiment poses two additional challenges.
The first challenge comes from the fact that, for this example,  it is not possible to derive and run additional instances.
We simulated the creation of new instances by reading only part of provenance and testing the algorithms on unread data, with an early stop when the pipeline instance to be tested was not present.

The second challenge was the number of properties -- a total of 202 numerical statistics. We applied feature selection and aggregated the values in buckets in order to increase the probability of  configurations that share parameter-value combinations. This reduced the configuration space to 15 parameters with 8 possible values (buckets) each. 
Since we were dealing with historical data, the instance execution time here is negligible.

We split the dataset into three parts: 50\% of the data was used for training; %
25\% was the budget for pipeline instances that any sub-method of \ourapproach requested;  and we create a 25\% holdout to assess the accuracy of \ourapproach's  minimal root causes as a classifier to predict when a pipeline instance will fail. Precisely, if the pipeline instance is a superset of a minimal root cause, we predict failure. This method is accurate 98\% of the time, results that are comparable to those reported in~\cite{Bailis:2017:MPA:3035918.3035928}. Thus, \ourapproach achieves concise explanations of the bugs and high classification accuracy.

\paragraph{Quality Measures.}
The root causes identified for all aforementioned pipelines were manually investigated to assess their soundness and to create ground truth for the real-world data. The ground truth allowed us to compute precision and recall and to compare with Data X-Ray and Explanation Tables. 
\begin{figure}[ht]
    \begin{center}
	    \includegraphics[width=0.87\columnwidth ]{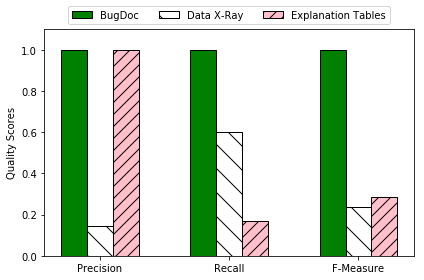}
	\end{center}
	\vspace{-.3cm}
	\caption[]{Real-World pipelines. \ourapproach (using \stacked and \debuggingdecisiontrees combined), outperforms Data X-Ray and Explanation tables.
	}
	\label{fig:polygamy_plot}
\end{figure}

\paragraph{Empirical Results.}
The recall metric in Figure~\ref{fig:polygamy_plot} shows that \ourapproach methods found all the parameter-comparator-value triples that would cause the execution of the pipelines to fail. As in Section~\ref{sec:synthetic}, Data X-Ray sometimes produces spurious root causes, %
yielding lower precision. By contrast, Explanation Tables shows high precision, but low recall.

\section{Conclusion}\label{sec:conclusion}

To the best of our knowledge, \ourapproach is the first method that autonomously finds minimal definitive root causes in computational pipelines or workflows. \ourapproach achieves this by analyzing previously executed computational pipeline instances, selectively executing new pipeline instances, and finding minimal explanations. 

When each root cause is due to a  single parameter-value setting or a single conjunction of parameter-equality-values, the shortcut methods of \ourapproach can provably guarantee to find at least one root cause in time proportional to the number of parameters (rather than exponential in the number of parameters as required by exhaustive search). Further, the shortcut approaches are guaranteed to find at least a subset of the parameter-values constituting a root cause in time linear in the number of parameters. When there are few parameters or sufficient 
computation time, the \debuggingdecisiontrees method of \ourapproach performs best.

Compared to the state of the art, \ourapproach makes no statistical assumptions (as do Bayesian optimization approaches like SMAC), but generally achieves better precision and recall given the same number of pipeline instances. %
In all cases, \ourapproach dominates the other methods based on the F-measure, though it may sometimes lose based on precision or recall individually. %
\ourapproach parallelizes well: pipeline instances can be executed in parallel, thus opening up the possibility of exploring large parameter spaces.

There are two main avenues we plan to pursue in future work. First, we would like to make \ourapproach available on a wide variety of provenance systems that support pipeline execution to broaden its applicability. 
Second, we would like to explore group testing~\cite{Lee2015,Macula2004} to identify problematic data elements when a dataset has been identified as a root cause. 
Another potential direction is the inclusion of observed variables (or predicates), properties that cannot be manipulated. While these cannot be used for deriving new instances, they can help enrich the explanations.

\paragraph{Acknowledgments.} We thank Data X-Ray and Explanation Tables authors for sharing their code with us. We are also grateful to Fernando Chirigati, Neel Dey, and Peter Bailis for providing the real-world pipelines. This work has been supported in part by NSF grants MCB-1158273, IOS-1339362, and MCB-1412232, CNPq (Brazil) grant 209623/2014-4, the DARPA D3M program, and NYU WIRELESS. Any opinions, findings, and conclusions or recommendations expressed in this material are those of the authors and do not necessarily reflect the views of funding agencies.

\bibliographystyle{ACM-Reference-Format}
\bibliography{paper}
\end{document}